\documentclass[12pt,reqno]{article}

\usepackage[usenames]{color}
\usepackage{amssymb}
\usepackage{amsmath}
\usepackage{amsthm}
\usepackage{amsfonts}
\usepackage{amscd}
\usepackage{graphicx}
\usepackage{mathtools}

\usepackage[colorlinks=true,
linkcolor=webgreen,
filecolor=webbrown,
citecolor=webgreen]{hyperref}

\definecolor{webgreen}{rgb}{0,.5,0}
\definecolor{webbrown}{rgb}{.6,0,0}

\usepackage{color}
\usepackage{fullpage}
\usepackage{float}

\usepackage{graphics}
\usepackage{latexsym}
\usepackage{epsf}
\usepackage{breakurl}

\usepackage{booktabs,makecell}
\usepackage{siunitx} 

\def\Que{\mathbb{Q}}

\newcommand{\seqnum}[1]{\href{https://oeis.org/#1}{\rm \underline{#1}}}

\begin{document}

\theoremstyle{plain}
\newtheorem{theorem}{Theorem}
\newtheorem{corollary}[theorem]{Corollary}
\newtheorem{lemma}[theorem]{Lemma}
\newtheorem{proposition}[theorem]{Proposition}

\newtheorem{definition}[theorem]{Definition}
\theoremstyle{definition}
\newtheorem{example}[theorem]{Example}
\newtheorem{conjecture}[theorem]{Conjecture}

\theoremstyle{remark}
\newtheorem{remark}[theorem]{Remark}

\makeatletter
\g@addto@macro\bfseries{\boldmath}
\makeatother

\title{Using finite automata to compute the base-$b$ representation of the golden ratio and other quadratic irrationals}

\author{
Aaron Barnoff and Curtis Bright \\
School of Computer Science \\
University of Windsor \\
Windsor, ON N9B 3P4 \\
Canada \\
\href{mailto:barnoffa@uwindsor.ca}{\tt barnoffa@uwindsor.ca} \\
\href{mailto:cbright@uwindsor.ca}{\tt cbright@uwindsor.ca}\\
\ \\
Jeffrey Shallit\\
School of Computer Science\\
University of Waterloo \\
Waterloo, ON  N2L 3G1 \\
Canada\\
\href{mailto:shallit@uwaterloo.ca}{\tt shallit@uwaterloo.ca} }

\maketitle

\begin{abstract}
We show that the $n$'th digit of the base-$b$ representation of
the golden ratio is a finite-state function of the
Zeckendorf representation of $b^n$, and hence can be computed
by a finite automaton.   Similar results can be proven for any quadratic irrational.   We use a satisfiability (SAT)
solver to prove, in some cases, that the
automata we construct are minimal.
\end{abstract}

\section{Introduction}

The base-$b$ digits of famous irrational numbers have been of
interest for hundreds of years.   For example, William Shanks
computed 707 decimal digits of $\pi$ in 1873 (but only the first
528 were correct) \cite{Shanks:1873}.   As a high school student, the third author
used a computer in 1976 to determine the first $10{,}000$ digits of the decimal
representation of $\varphi = (\sqrt{5}+1)/2$, the golden ratio, using the computer language
APL \cite{Shallit:1976}.

The celebrated results of 
Bailey, Borwein, and Plouffe \cite{Bailey&Borwein&Plouffe:1997}
demonstrated that one can compute the $n$'th bit of certain
famous constants, such as $\pi$, in very small space.\footnote{Sometimes
this result is described as ``computing the $n$'th digit without
having to compute the previous $n-1$ digits''. But this is not
really a meaningful assertion, since the phrase
``computing $x$ without computing $y$'' is not so well-defined.}

{\it Can finite automata generate the base-$b$ digits of irrational algebraic numbers, such as $\varphi$?}   This fundamental question was raised by Cobham in the late 1960's (a re-interpretation of a related question due to Hartmanis and Stearns \cite{Hartmanis&Stearns:1965}). Though Cobham believed for a time that he had proved they cannot be so generated \cite{Cobham:1968b},    his proof was flawed, and it was not until 2007 that Adamczewski and Bugeaud \cite{Adamczewski&Bugeaud:2007} succeeded in proving that there is no deterministic finite automaton with output that, on input $n$ expressed in base $b$, returns the $n$'th base-$b$ digit of an irrational real algebraic number $\alpha$.

Even so, in this paper we show that, using finite
automata, one {\it can\/} compute the $n$'th digit in the
base-$b$ representation of the golden ratio $\varphi$! At first glance this might seem to contradict the Adamczewski--Bugeaud result.
But it does not, since for our theorem the input is not $n$ 
expressed in base $b$, but
rather $b^n$ in an entirely different numeration system, the
Zeckendorf representation.   Analogous results exist for any quadratic
irrational.

Our result does not give a particularly efficient way to
compute the base-$b$ digits of quadratic irrationals, but it is nevertheless somewhat surprising.  
Using a SAT solver, in some cases (such for the binary digits
of $\varphi$)
we prove that the automata constructed is minimal and unique.
Interestingly, in other cases (such as for the ternary digits of $\varphi$)
we were able to prove the minimality of our automaton, but we
discovered several distinct automaton with the same number of states
computing the same quadratic irrational (at least up to a very high precision).
It is conceivable that the automata produced by our method are indeed
minimal and unique in general, and we leave this as an open question.

\section{Number representations and automata}
A DFAO (deterministic finite automaton with output)
$A$ consists
of a finite number of states, and labeled transitions connecting
them.  The automaton processes an input string $x$ by starting in the
distinguished start state $q_0$, and then following the transitions from state to
state, according to each successive bit of $x$.   Each state $q$ has an output $\tau(q)$ associated with it, and the function $f_A$ computed by the DFAO maps the input
$x$ to the output associated with the last state reached.  For an example of a DFAO, see Figure~\ref{fig2}.

A DFA (deterministic finite automaton) is quite similar to a DFAO.
The only difference is that there are exactly two possible outputs
associated with each state, either $0$ or $1$.   States with an
output of $1$ are called ``accepting'' or ``final''.  If an input
results in an output of $1$, it is said to be accepted by the DFA.
A {\it synchronized\/} DFA \cite{Carpi&Maggi:2001} is a particular type of DFA that takes two
inputs in parallel; this is accomplished by making the input alphabet
a set of pairs of alphabet symbols.  A synchronized automaton computes
a synchronized sequence $(f(n))_{n \geq 0}$; it does this by accepting exactly the inputs
where the first components spell out a representation of $n$, and
the second components spell out a representation for $f(n)$, where leading zeros may be required to make the inputs the same length.  For more
about synchronized sequences, see \cite{Shallit:2021h}.  An example of a synchronized DFA appears in Figure~\ref{fig1}.

Let $x$ be a non-negative real number, and write its base-$b$
representation in the form $x = \sum_{-\infty < i \leq t} a_i b^i = a_t a_{t-1} \cdots a_0 . a_{-1} a_{-2} \cdots $,
where $a_i \in \{0,1,\ldots, b-1 \}$.  For $n \geq 0$, we call $a_{-n-1}$ the
$n$'th digit to the right of the point.   The indexing is perhaps a little
unusual, but it seems to decrease the size of the automata produced.   

\subsection{Zeckendorf representation}
The Fibonacci numbers are defined, as usual, by $F_0 = 0$,
$F_1 = 1$, and $F_n = F_{n-1} + F_{n-2}$.    The Zeckendorf
representation \cite{Lekkerkerker:1952,Zeckendorf:1972}
of a natural number $n$ is the unique way
of writing $n$ as a sum of Fibonacci numbers $F_i$, $i \geq 2$,
subject to the condition that no two consecutive Fibonacci
numbers are used.   We may write the Zeckendorf representation
as a binary string $(n)_F = a_1 \cdots a_t$, where
$n = \sum_{1 \leq i \leq t} a_i F_{t+2-i}$.   For example,
$(43)_F = 34 + 8 + 1 = F_9 + F_6 + F_2 $ has representation
$10010001$.  The substring $11$ cannot occur
due to the rule that two consecutive Fibonacci
numbers cannot be used.
In what follows, leading zeros in strings are typically
ignored without comment.

To illustrate these ideas,
Table~\ref{table3} gives the Zeckendorf representation of the
first few powers of $2$ and $3$.   We will use them in
Section~\ref{phi23}.
\setlength{\tabcolsep}{5pt}
\begin{table}[htb]
    \centering
    \begin{tabular}{c|r|r|r|r}
    $i$ & $2^i$ & $(2^i)_F$ & $3^i$ & $(3^i)_F$ \\ 
    \hline
0 & 1 &  1    &      1 & 1          \\  
1 & 2 & 10    &    3  &  100         \\ 
2 & 4 &  101  &     9  & 10001        \\
3 & 8 & 10000  &    27 & 1001001     \\ 
4 & 16 &  100100 &   81 & 101001000  \\  
5 & 32 & 1010100  &  243 & 100000010010 \\
    \end{tabular}
    \caption{Zeckendorf representation of the first few powers of $2$ and $3$.}
    \label{table3}
\end{table}

\section{Automata and the base-$b$ representation of $\varphi$}
\label{phi23}

Our main result is Theorem~\ref{main} below.
\begin{theorem}
For all integers $b \geq 2$, there exists a DFAO $A_b$ that,
on input the Zeckendorf representation of $b^n$, computes
the $n\mspace{-2mu}$'th digit to the right of the point
in the base-$b$ representation of $\varphi$.
\label{main}
\end{theorem}

\begin{proof}
It is known that there exists a $7$-state
synchronized DFA accepting, in parallel, the
Zeckendorf representations of $n$ and $\lfloor n \varphi \rfloor$
for all $n \geq 0$~\cite[Thm.~10.11.1(a)]{Shallit:2023}.
Its transition diagram is depicted in Figure~\ref{fig1}, where accepting 
states are denoted by double circles, and $0$ is the initial state, labeled by a headless arrow entering.   

The DFA is constructed using the fact that $\lfloor n \varphi \rfloor = [(n-1)_F 0]_F + 1$, where $[(n-1)_F 0]_F$ is the left shift of the string $(n-1)_F$.  
For example, $\lfloor 11\varphi \rfloor = 17$, and to determine $11 \rightarrow 17$, we find $(10)_F = 10010$, left-shift that to get $100100=(16)_F$, and add $1$ to get $17$. 
\begin{figure}[htb]
\begin{center}
\includegraphics[width=6in]{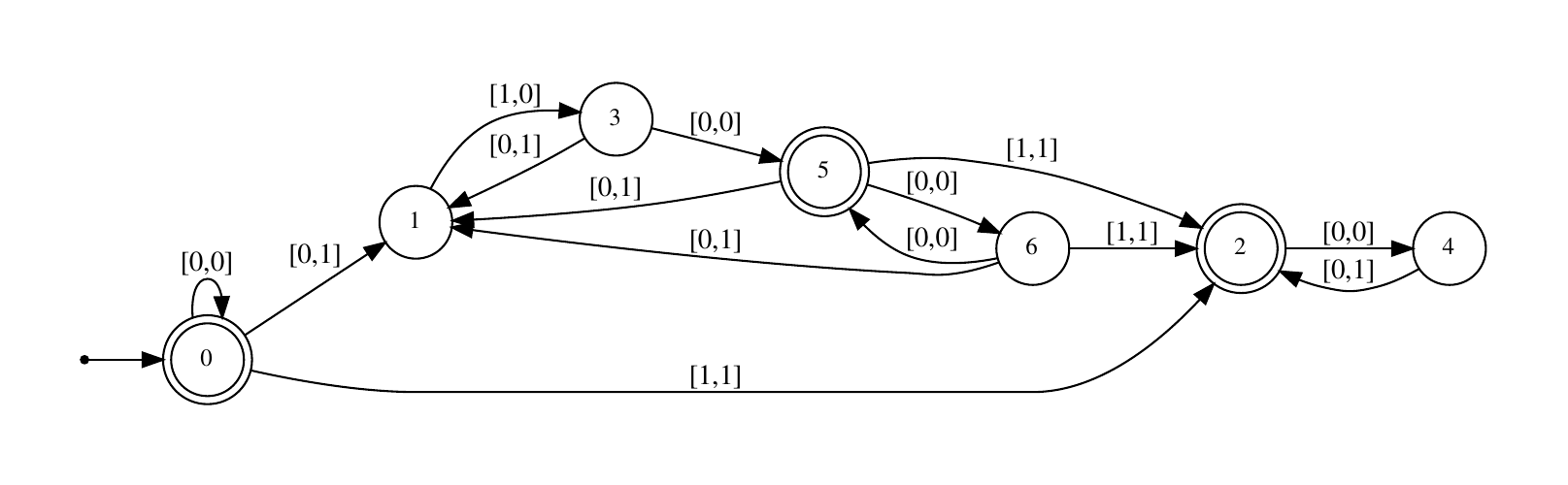}
\end{center}
\caption{Synchronized automaton for $\lfloor n \varphi \rfloor$. The inputs are the Zeckendorf representation of $n$ and $x$, in parallel.}
\label{fig1}
\end{figure}
To understand how to use this automaton, observe
that $(11)_F = 10100$ and $\lfloor 11 \varphi \rfloor = 17$ and
$(17)_F = 100101$.   Since these two numbers have representations
of different lengths, we need to pad the former with a leading $0$.
Then if $x = [0,1][1,0][0,0][1,1][0,0][0,1]$, the first components
concatenated spell out $010100$ and the second components
spell out $100101$.   When we input this, we visit, successively,
states $1,3,5,2,4,2$, and so we accept.

Let $x$ be a positive real number, with base-$b$ representation
$y.a_0 a_1 a_2 \cdots$, where the period is the analogue of the
decimal point for base $b$, and $y$ is an arbitrary finite block of
digits.
Now $b^{n+1} x$ has base-$b$ representation
$y a_0 a_1 \cdots a_{n-1} a_n. a_{n+1} \cdots$ and
$\lfloor b^{n+1} x \rfloor$ has base-$b$ representation
$y a_0 a_1 \cdots a_{n-2} a_{n-1} a_n$.
Similarly, $b \lfloor b^n x \rfloor$ has base-$b$
representation
$y a_0 a_1 \cdots a_{n-1} 0$.
Hence $\lfloor b^{n+1} x \rfloor - b \lfloor b^n x \rfloor = a_n$.
In the particular case where $x = \varphi$, we get a formula for the $n$'th digit to the right of the decimal point of $\varphi$, namely
\begin{equation*}
D_b (n) \coloneqq \lfloor b^{n+1} \varphi \rfloor - b \lfloor b^n \varphi \rfloor .
\label{dbi}
\end{equation*}

From the DFA computing $\lfloor n \varphi \rfloor$,
it is possible to create another DFA
accepting, in parallel, the Zeckendorf representations of
$q$ and $\lfloor bq \varphi \rfloor - b \lfloor q \varphi \rfloor$.
This is based on the fact that there is an algorithm to compile
a first-order logic statement involving the usual logical
operations (AND, OR, NOT, etc.), the integer
operations of addition, subtraction, multiplication by constants, and the universal and existential quantifiers, into an automaton that accepts the Zeckendorf representation of those integers making the statement true \cite{Mousavi&Schaeffer&Shallit:2016}.

From this DFA, we can compute $b$ individual DFAs $A_{b,i}$ accepting the Zeckendorf representation of those $n$ for which $\lfloor bq \varphi \rfloor - b \lfloor q \varphi \rfloor = i$, for $0 \leq i < b$.
Finally, we combine all the $A_{b,i}$ together into a single DFAO
(using a product construction for automata) computing the difference
$\lfloor bq \varphi \rfloor - b \lfloor q \varphi \rfloor$.

By substituting $q = b^n$, we see that this automaton is the
desired one, computing $D_b (n)$ on input the Zeckendorf representation of $b^n$. 
\end{proof}

We now use {\tt Walnut}, which is free software for compiling first-order logical
expressions into automata, to explicitly compute the automata for $\varphi$ in base $2$ and base $3$.
For base $2$, we need the following
{\tt Walnut} commands:
\begin{verbatim}
reg shift {0,1} {0,1} "([0,0]|[0,1][1,1]*[1,0])*":
def phin "?msd_fib (s=0 & n=0) | Ex $shift(n-1,x) & s=x+1":
def phid2 "?msd_fib Ex,y $phin(2*n,x) & $phin(n,y) & x=2*y+1":
combine FD2 phid2:
\end{verbatim}
These produce the DFAO in Figure~\ref{fig2}.
\begin{figure}[htb]
\begin{center}
\includegraphics[width=6in]{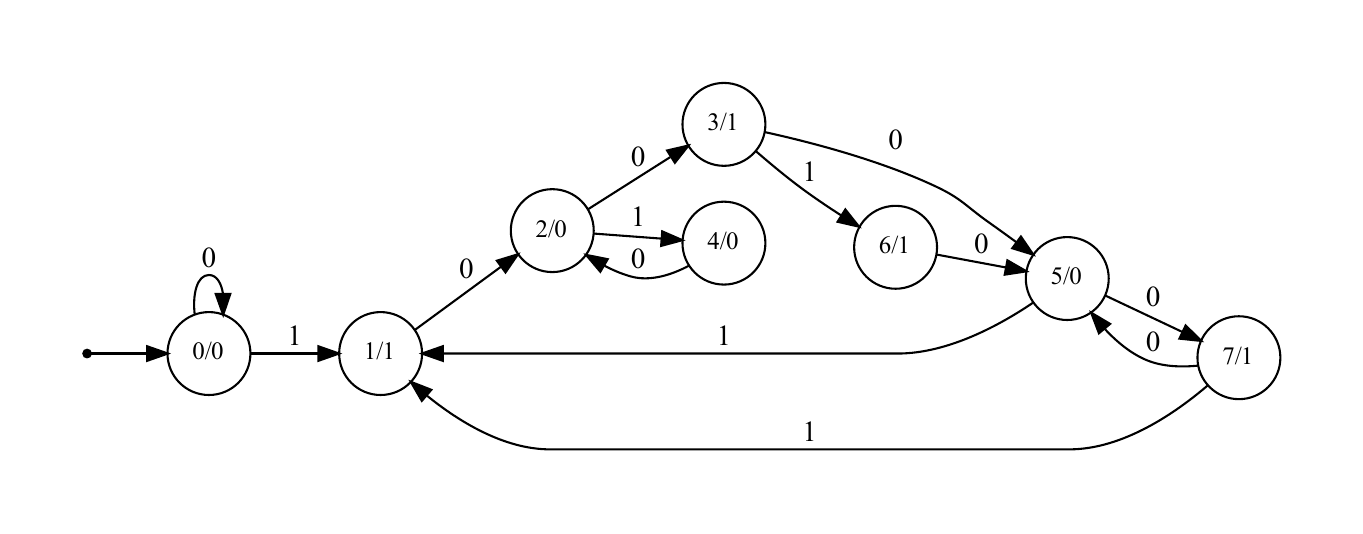}
\end{center}
\caption{Automaton for the $n$'th bit to the right of the binary point
of $\varphi$.
States are labeled in the form $a/b$, where $a$ is the state number and $b$ is the output.
The input is the Zeckendorf representation of $2^n$, and the output is $b$ when the last state reached is labeled $a/b$.}
\label{fig2}
\end{figure}

For example, in base $2$, we have
$\varphi = 1. 1 0 0 1 1 1 1 0 0 0 1 1 0 1 1 1 \cdots$.
To compute the 4th digit to the right of the binary point
we write $2^4 = 16$ in Zeckendorf representation,
namely $100100$, and feed it into the automaton,
reaching states $1,2,3,6,5,7$ successively,
with output $1$ at the end.  

We now explain the {\tt Walnut} commands.   The first line creates the DFA {\tt shift}, using a regular expression; it takes two base-$2$ inputs and only accepts if the second input is the left shift of the first.  Next is the DFA {\tt phin}, which is shown in Figure~\ref{fig1} and uses {\tt shift} to check that its two inputs have the relationship $(n)_F$ and $[(n-1)_F 0]_F + 1$, which computes the function $n \rightarrow \lfloor n \varphi \rfloor$ in a synchronized fashion. Next, the DFA {\tt phid2}, when given the representation of $q$ as input, 
accepts if $\lfloor 2q \varphi \rfloor -  2\lfloor q \varphi \rfloor = 1$, and rejects otherwise.
Lastly, {\tt combine}  converts {\tt phid2} into a DFAO by replacing its accepting and rejecting states with output value $1$ and $0$, respectively. 

The automaton for base $3$ (see Figure~\ref{fig3}) can be constructed similarly with the following {\tt Walnut} commands:
\begin{verbatim}
reg shift {0,1} {0,1} "([0,0]|[0,1][1,1]*[1,0])*":
def phin "?msd_fib (s=0 & n=0) | Ex $shift(n-1,x) & s=x+1":
def phid3a "?msd_fib Ex,y $phin(3*n,x) & $phin(n,y) & x=3*y+1":
def phid3b "?msd_fib Ex,y $phin(3*n,x) & $phin(n,y) & x=3*y+2":
combine FD3 phid3a=1 phid3b=2:                   
\end{verbatim}
\begin{figure}[htb]
\begin{center}
\includegraphics[width=6in]{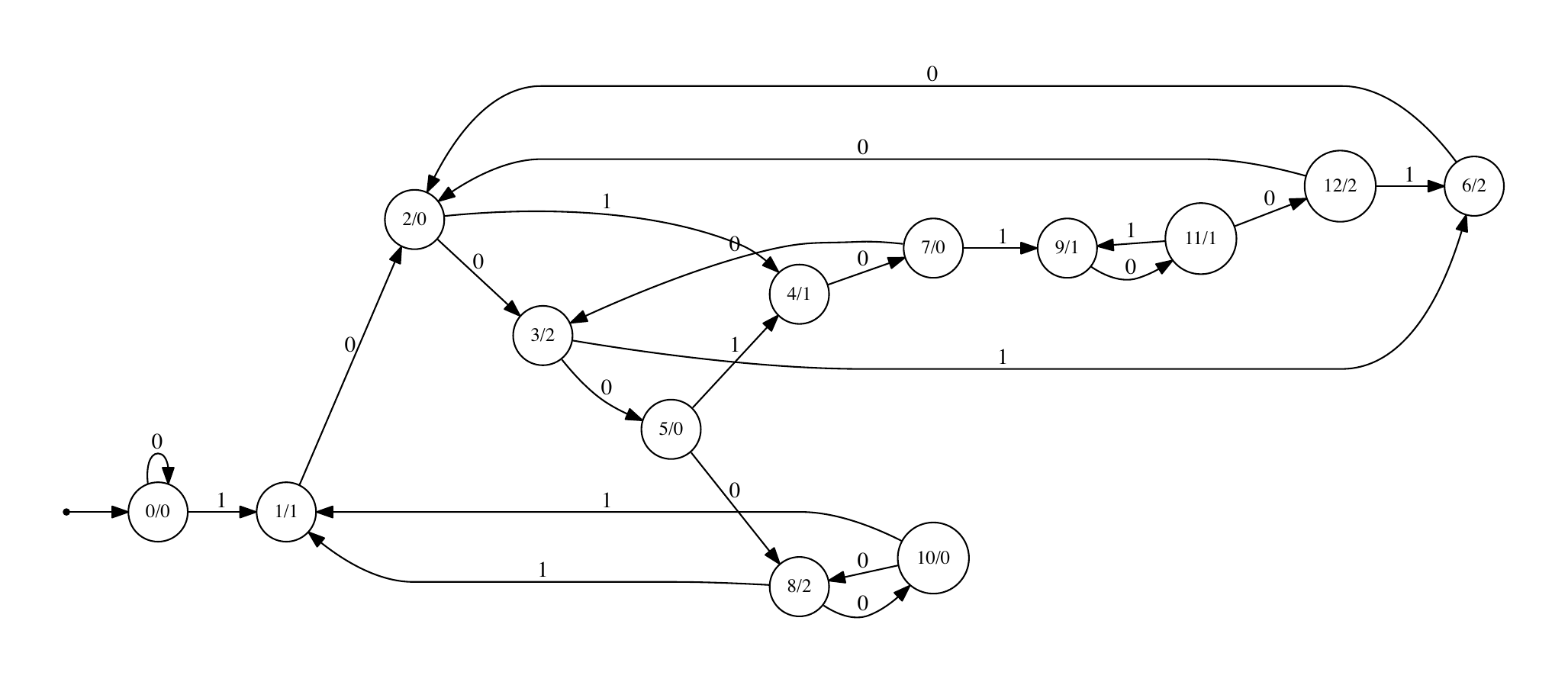}
\end{center}
\caption{Automaton for the $n$'th bit to the right of the point of $\varphi$ in base $3$.}
\label{fig3}
\end{figure}

In base $3$, 
$\varphi = 1. 1 2 1 2 0 0 1 1 2 2 0 2 1 2 1 0 \cdots$. To compute the $3$rd digit to the right of the point we write $3^3 = 27$ in Zeckendorf representation as $1001001$ and pass it to the automaton, which traverses states $1,2,3,6,2,3,6$ successively, giving an output of $2$. 


There is no conceptual barrier to carrying out similar computations for any
base $b\geq 2$.  For base $10$, for example, {\tt Walnut} computes a
finite automaton with $97$ states that, on input $(10^n)_F$, returns
the $n$'th digit to the right of the decimal point in the
decimal expansion of $\varphi$.  The transition diagram is, unfortunately, too
complicated to display here, but the transition table and output function is given in the Appendix.

\section{Other quadratic irrationals}
\label{sec:otherQI}

There is nothing special about $\varphi$, and the same ideas
can be used for any quadratic irrational, although
the input representation requires some modification.

\subsection{Pell representation}
\label{sec:Pell}

Another representation for the natural numbers is based on the Pell numbers, defined by $P_0 = 0$, $P_1 = 1$, and $P_n = 2P_{n-1} + P_{n-2}$ for $n \geq 2$.  We can then write every natural number
$n = \sum_{1 \leq i \leq t} a_i P_{t+1-i}$ where
$a_i \in \lbrace 0,1,2 \rbrace$.   To get uniqueness of the representation, we have to impose two conditions.  First, we must have that
$a_t \not= 2$.  Second, if $a_i = 2$, then $a_{i+1} = 0$.  
See \cite{Baranwal&Shallit:2019} for more details.   The unique representation, over the alphabet $\{0,1,2\}$, is denoted $(n)_P$.

Table~\ref{table4} gives the Pell representation of the
first few powers of $2$ and $3$.   We will use them in
Section~\ref{sec:Pell}.
\begin{table}[htb]
    \centering
    \begin{tabular}{c|r|r|r|r}
    $i$ & $2^i$ & $(2^i)_P$ & $3^i$ & $(3^i)_P$ \\ 
    \hline
0 & 1 &  1    &      1 & 1          \\  
1 & 2 & 10    &    3  &  11         \\ 
2 & 4 &  20  &     9  & 120        \\
3 & 8 & 111  &    27 & 2011     \\ 
4 & 16 &  1020 &   81 & 100201  \\  
5 & 32 & 10011  &  243 & 1100020 \\
    \end{tabular}
    \caption{Pell representation of the first few powers of $2$ and $3$.}
    \label{table4}
\end{table}

The Pell numeration system in {\tt Walnut} can be used to construct
automata computing the base-$b$ digits of $\sqrt{2}$, just as we did for $\varphi$.
This results in a 6-state DFAO for base 2 (see Figure~\ref{fig:sqrt2}), and a 14-state DFAO for base 3.
The {\tt Walnut} commands for base 2 are:
\begin{verbatim}
reg pshift {0,1,2} {0,1,2} 
   "([0,0]|([0,1][1,1]*([1,0]|[1,2][2,0]))|[0,2][2,0])*":
def sqrt2n "?msd_pell (s=0 & n=0) | Ex $pshift(n-1,x) & s=x+2":
def sqrt2d2 "?msd_pell Ex,y $sqrt2n(2*n,x) & $sqrt2n(n,y) 
   & x=2*y+1":
combine SD2 sqrt2d2:
\end{verbatim}
\begin{figure}[htb]
\begin{center}
\includegraphics[width=6in]{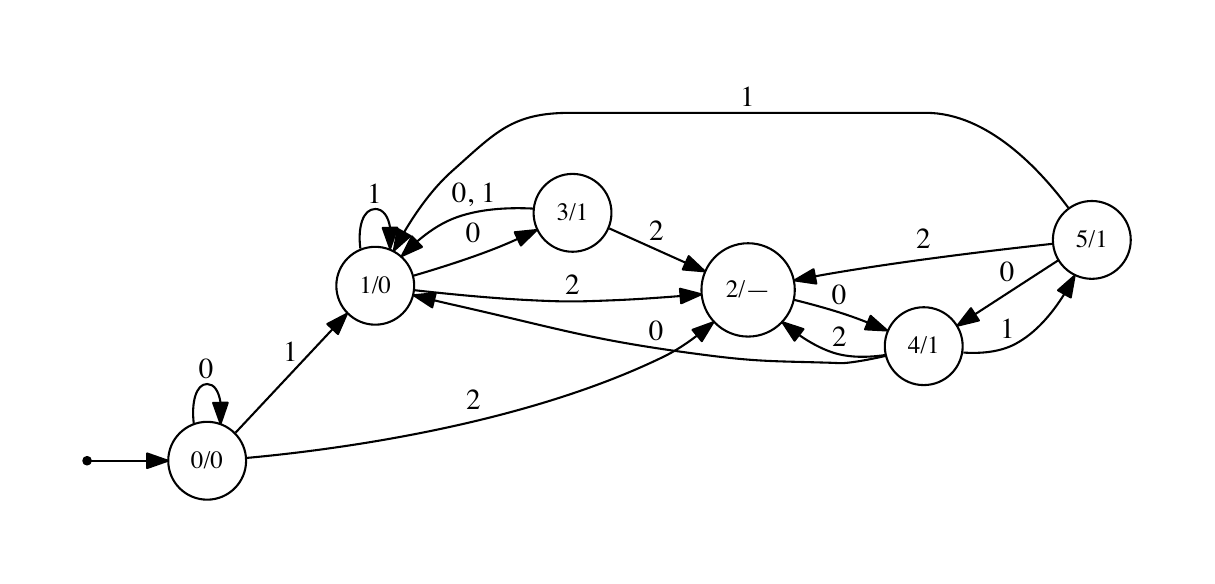}
\end{center}
\caption{Automaton for the $n$'th bit to the right of the binary point
of $\sqrt{2}$.   Input is $2^n$ in Pell representation.}
\label{fig:sqrt2}
\end{figure}
The alert reader will observe that no output is associated with
state $2$.   This is because inputs that lead to this state, such as $12$, are not valid Pell representations.  However, the state cannot be removed, because $120$ {\it is\/} a valid Pell representation.

\subsection{Ostrowski representation}
\label{sec:OstNum}
Of course, what makes our results work is that the numeration systems are ``tuned'' to the particular quadratic irrational we want to compute.   For $\varphi$, this is the Fibonacci numbers; for $\sqrt{2}$, the Pell numbers.   We need to find an appropriate numeration system that is similarly ``tuned'' to any quadratic irrational.   It turns out that the proper system is
the Ostrowski numeration system \cite{Baranwal&Schaeffer&Shallit:2021,Ostrowski:1922}.

Every irrational real number $\alpha$ can be expressed uniquely as an infinite simple continued fraction $\alpha = [d_0, d_1, d_2, \ldots]$. Furthermore, $q_n$ is called the denominator of a convergent for $\alpha$ if $q_{-2}=1$, $q_{-1}=0$, and $q_n=d_nq_{n-1}+q_{n-2}$ for $n\geq 0$. For example, the continued fraction for $\pi$ is $[3, 7, 15, 1, \ldots]$, corresponding to the sequence $(q_i)_{i\geq 0} = 1, 7, 106, 113, \dotsc$.

The Ostrowski $\alpha$-numeration system uses the sequence $(q_n)_{n\geq0}$ of the denominators of the convergents for $\alpha$ to construct a unique representation for a non-negative integer
$N$ expressed as
\[ N = [a_{n-1} a_{n-2} \cdots a_0]_\alpha = \sum_{0\leq i < n}a_iq_i , \]
where
\begin{enumerate}
  \setlength{\itemsep}{1pt}
  \setlength{\parskip}{0pt}
  \setlength{\parsep}{0pt}
\item $0 \leq a_0 < d_1$;
\item $0 \leq a_i \leq d_{i+1}$ for $i \geq 1$; and
\item for $i \geq 1$, if $a_i = d_{i+1}$ then $a_{i-1} = 0$.
\end{enumerate}

The Ostrowski $\alpha$-representation for $N = [a_{n-1} a_{n-2} \cdots a_0]_\alpha$ is then determined with a greedy algorithm, starting at the most significant term and choosing the largest multiple $a_{n-1}$ for $q_{n-1}$ that is less than $N$, and then applying the same algorithm to $N-a_{n-1} q_{n-1}$. For example, for $\alpha = \sqrt{3}+1 = [2, \overline{1,2}]$, the denominators of the continued fraction convergents form the sequence $(q_i) = 1, 1, 3, 4, 11, 15, \dotsc$ (OEIS \seqnum{A002530}). Rule $1$ for the construction forces $a_0 = 0$ because $d_1 = 1$, while rule $2$ requires that $a_1 \leq d_2 = 2$, $a_2 \leq d_3 = 1,$ and so on. Rule $3$ ensures uniqueness by enforcing the constraint that if $a_1 = d_2 = 2$, then $a_2 = 0$, and if $a_2 = d_3 = 1$, then $a_3 = 0$, and so on. Then, for example, the $\alpha$-representation of $37 = 2 \cdot 15 + 4 + 3 = 2q_5 + q_3 + q_2 = [20110]_\alpha$.

In order to construct a DFAO $A_b$ that, given the input of the Ostrowski $\alpha$-representation of $b^n$, computes the $n$'th digit to the right of the point in the base-$b$ representation of $\alpha$, 
we require an Ostrowski $\alpha$-synchronized function $n \rightarrow \lfloor n\alpha \rfloor$. 
It was shown in \cite{Schaeffer&Shallit&Zorcic:2024} that every quadratic irrational  $0 < \beta < 1$ with a purely periodic continued fraction $[0, \overline{ d_1, d_2, \ldots, d_m}]$ has an Ostrowski $\beta$-synchronized sequence $(\lfloor n \beta \rfloor)_{n\geq 1}$, such that
\begin{equation}
[(n-1)_\beta 0^m]_\beta = q_m(n-1) + q_{m-1} \cdot \lfloor n\beta \rfloor,
\label{eq:beatty}
\end{equation}
where $q_i$ is the denominator of the $i$'th convergent to $\beta$, and $(n-1)_\beta0^m$ is the $\beta$-representation of $n-1$, left-shifted $m$ times.  

Furthermore, it was shown that if  $\alpha > 0$ belongs to $ \Que(\beta)$, then $(\lfloor n \alpha \rfloor)_{n\geq 1}$ is synchronized in terms of the Ostrowski $\beta$-representation through the relation $\alpha = (a+b\beta)/c$, where $b$,~$c \geq 1$, and
\begin{equation}
    \lfloor  n\alpha  \rfloor = \left\lfloor \frac{\lfloor bn \beta \rfloor + an}{c} \right\rfloor.
\label{eq:beatty2}
\end{equation}
This is notable because when constructing an Ostrowski $\alpha$-representation with {\tt Walnut}, it is assumed that $0 < \alpha < \frac{1}{2}$, which corresponds to a continued fraction with terms $d_0 = 0$ and $d_1 > 1$. If $\alpha \geq \frac{1}{2}$, then we can set $d_0 = 0$ and rotate the period until $d_1 >1$, giving a quadratic irrational $0 < \beta < \frac{1}{2}$ corresponding to the periodic part of $\alpha$. Then an Ostrowski representation for $\beta$ can be constructed, and Eq.~\eqref{eq:beatty} is used to find an automaton for $\lfloor n \beta \rfloor$, followed by Eq.~\eqref{eq:beatty2} to find an automaton for $\lfloor n \alpha \rfloor$. Therefore, $(\lfloor n \alpha \rfloor)_{n\geq1}$ is synchronized in terms of the Ostrowski $\beta$-representation. 

For example, for $\alpha = \sqrt{3}+1 = [2, \overline{1,2}]$, we have $\alpha \geq \frac{1}{2}$.  Since
we only compute the digits after the decimal point, we set $d_0 = 0$ and then rotate the period to get $\beta = [0, \overline{2,1}] = (\sqrt{3}-1)/2$. This gives the sequence of denominator convergents $1, 2, 3, 8, 11, 30, \dotsc$, where $m=2$, $q_m = 3$, and $q_{m-1} =2$, and so Eq.~\eqref{eq:beatty} gives $[(n-1)_\beta 00]_\beta = 3\lfloor n \beta \rfloor + 2(n-1)$. This results in a DFA for $\lfloor n \beta \rfloor$ that has 23 states. Then, we find $\alpha = (2+2\beta)/1$, with $a=2$, $b=2$, and $c=1$, and Eq.~\eqref{eq:beatty2} gives a DFA with 20 states, shown in Figure \ref{fig_pvan}. 

\begin{figure}
\begin{center}
\includegraphics[width=6in]{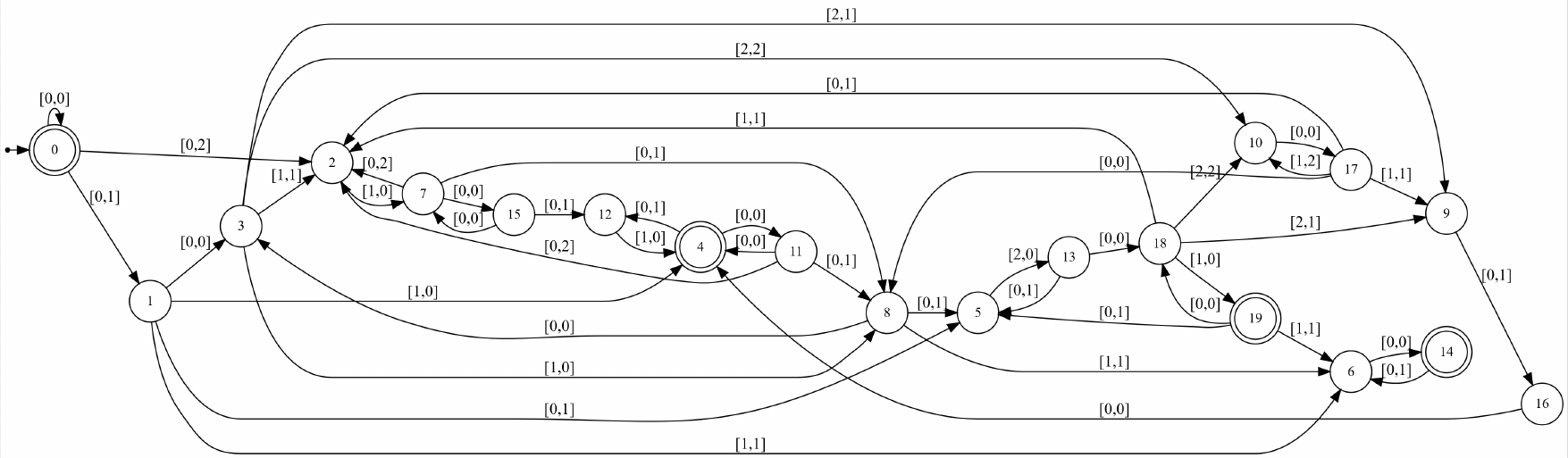}
\end{center}
\caption{Synchronized automaton for $\lfloor n \alpha \rfloor$ for $\alpha=\sqrt{3}+1$.}
\label{fig_pvan}
\end{figure}
Then, for example $(5)_\beta = 110$ and $\lfloor 5\alpha \rfloor = 13$ and $(13)_\beta = 10010$. When we input $[0,1][0,0][1,0][1,1][0,0]$ into the automaton, we visit states $1, 3, 8, 6, 14$ in succession, and so we accept. From here, the same process that is outlined in Section~\ref{phi23}
can be used to construct a DFA accepting in parallel the Ostrowski $\alpha$-representations of $q$ and $\lfloor bq\alpha \rfloor - b\lfloor q \alpha \rfloor$, and ultimately the DFAO $A_b$ as desired.

\subsection{{\tt Walnut} implementation}\label{sec:Walnut}
Constructing the DFAOs for other quadratic irrationals with {\tt Walnut} requires the {\tt ost} command to create custom Ostrowski representations. As explained above, {\tt Walnut} requires that $0 < \beta < \frac{1}{2}$ to create the corresponding Ostrowski representation, and it is possible to create a DFAO for $\alpha \geq \frac{1}{2}$ by synchronizing it in terms of the the Ostrowski representation for $\beta$. Presented below are the general steps for constructing a DFAO for a quadratic irrational $\alpha$ with {\tt Walnut}, using the process explained above with Equations \eqref{eq:beatty} and \eqref{eq:beatty2}.

First, we construct the continued fraction of $\beta < \frac{1}{2}$ from $\alpha$ by setting $d_0=0$ and rotating the period until $d_1 > 1$, if necessary. Next, we  determine the denominators $j = q_m$ and $k = q_{m-1}$ of the continued fraction convergent to $\beta$, where $m$ is the number of elements in the period. Lastly, we find $a$, $b$, and $c$ from the relation $\alpha = (a+b\beta)/c$, where $b,c \geq 1$. With these, we can use the following {\tt Walnut} commands:
\begin{verbatim}
# Construct Ostrowski representation for Beta
    ost ostBeta [0] [d1 d2 ... dm];
# Create a DFA of z = floor(n*Beta) using j and k
    def betan "?msd_ostBeta Eu,v n=u+1 & $shift(u,v) & v=k*z+j*u":
# Create a DFA of z = floor(n*Alpha) synchronized 
    def alphan "?msd_ostBeta Eu $betan(b*n,u) & z=(u+a*n)/c":
# Create a DFAO for Alpha in base 2
    def alphan_d2 "?msd_ostBeta Ex,y $alphan(2*n,x) & $alphan(n,y) 
    & x!=2*y":
    combine AD2 alphan_d2:
\end{verbatim}

The {\tt shift} DFA can be constructed from a regular expression as done above for $\varphi$,
and is based on the specific representation and continued fraction sequence. If multiple left-shifts are required, it may be simpler to create a {\tt shift} DFA that left-shifts only one position at a time, and chain its use together multiple times. For example, three left-shifts could be achieved using a 1-shift DFA by:
\begin{verbatim}
def betan "?msd_ostBeta Eu,v,w,x n=u+1 & $shift(u,v) & $shift(v,w) 
       & $shift(w,x) & x=k*z+j*u":
\end{verbatim}

Using this process, we created the DFAOs for other quadratic irrationals including
the ``bronze ratio'' $(\sqrt{13}+3)/2=[3,\overline{3}]$ and several Pisot numbers.

\subsection{\texttt{Walnut} code for quadratic irrationals}
\label{sec:Walnut-extra}

In this section we give \texttt{Walnut} code and provide the resulting automata for several
more quadratic irrationals, including the ``bronze ratio''
$(\sqrt{13}+3)/2=[3,\overline3]$ in bases 2 and 3.  We also
provide code that construct automata computing the Pisot numbers
$\sqrt3+1$ and $(\sqrt{17}+3)/2$ and some other closely related
quadratic irrationals.

\subsubsection{The bronze ratio $(\sqrt{13}+3)/2$ in base $2$ and base $3$}
\label{sec:bronze}
\begin{verbatim}
ost bt [0] [3];
reg bts {0,1,2,3} {0,1,2,3} 
   "([0,0]|[0,2][2,2]*[2,0]|([0,2][2,2]*[2,3]|[0,3])
   [3,0]|([0,1]|[0,2][2,2]*[2,1])([1,1]|[1,2][2,2]*[2,1])*
   (([1,2][2,2]*[2,3]|[1,3])[3,0]|[1,2][2,2]*[2,0]|[1,0]))*":
def btbn "?msd_bt Eu,v n=u+1 & $bts(u,v) & v=1*z+3*u":
def btan "?msd_bt Eu $btbn(1*n,u) & z=(u+3*n)/1":
\end{verbatim}
DFAO for the bronze ratio in base 2 (see Figure~\ref{figbrb2}):
\begin{verbatim}
def btn_d2 "?msd_bt Ex,y $btan(2*n,x) & $btan(n,y) & x!=2*y":
combine BTND2 btn_d2:
\end{verbatim}
DFAO for the bronze ratio in base 3 (see Figure~\ref{figbrb3}):
\begin{verbatim}
def btn_d3a "?msd_bt Ex,y $btan(3*n,x) & $btan(n,y) & x=3*y+1":
def btn_d3b "?msd_bt Ex,y $btan(3*n,x) & $btan(n,y) & x=3*y+2":
combine BTND3 btn_d3a btn_d3b:    
\end{verbatim}

\begin{figure}
\begin{center}
\includegraphics[width=4in]{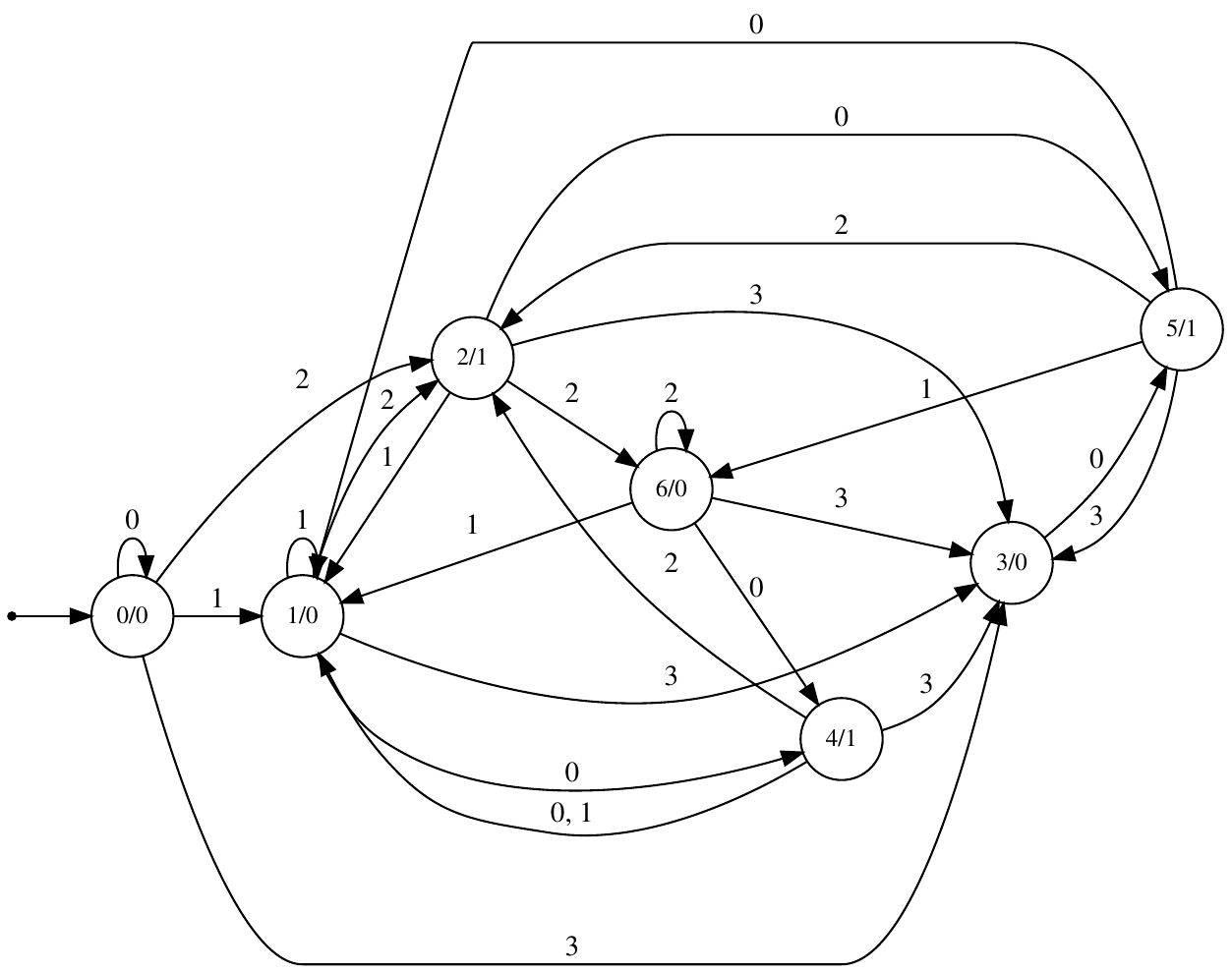}
\end{center}
\caption{Automaton for the $n$'th bit to the right of the point of 
$(\sqrt{13}+3)/2$
in base $2$.}
\label{figbrb2}
\end{figure}

\begin{figure}
\begin{center}
\includegraphics[width=4in]{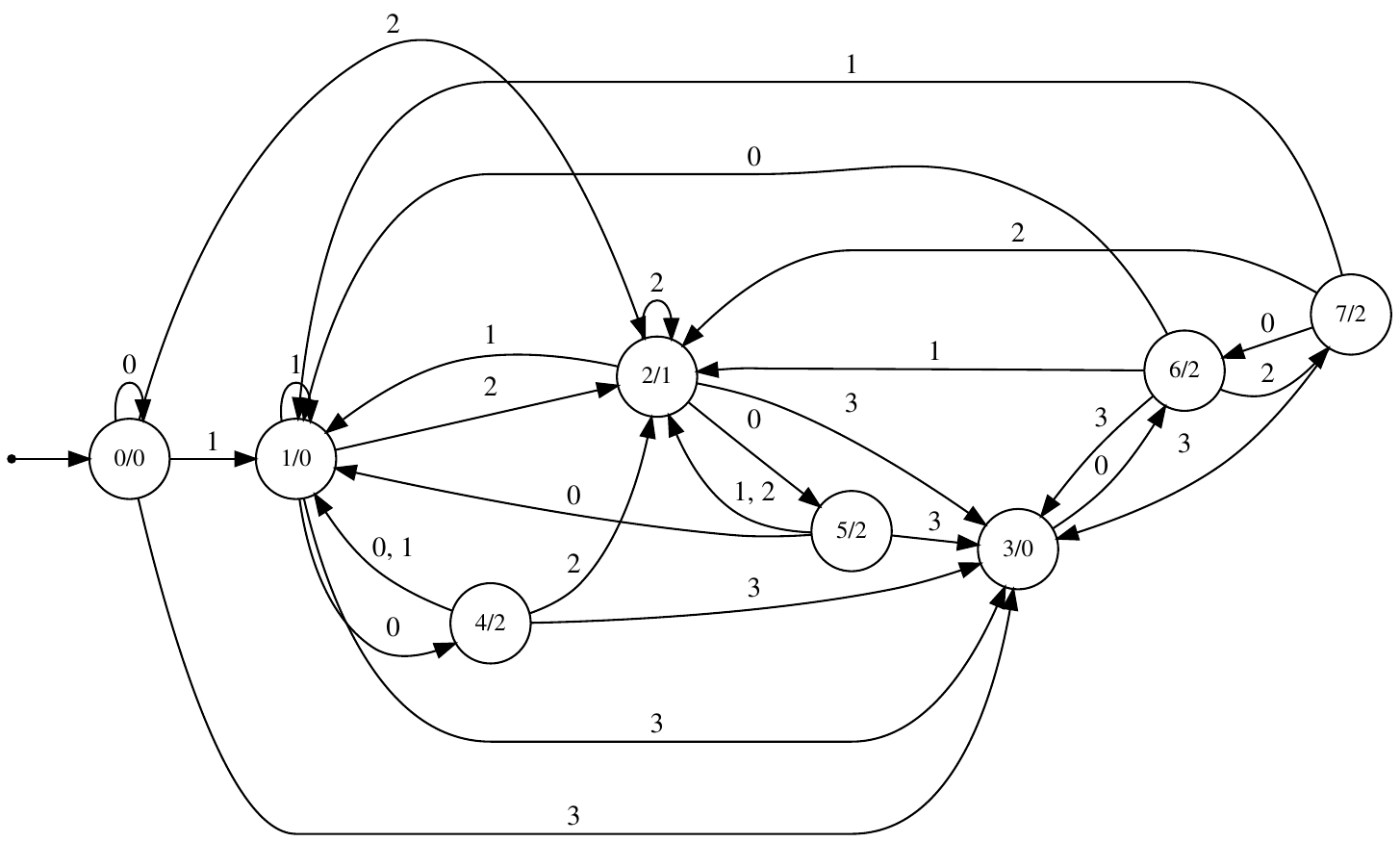}
\end{center}
\caption{Automaton for the $n$'th bit to the right of the point of 
$(\sqrt{13}+3)/2$
in base $3$.}
\label{figbrb3}
\end{figure}

\subsubsection{Pisot number $\sqrt{3}+1$ and $(\sqrt{3}-1)/2$ in base $2$}
\label{sec:Pisot1}
\begin{verbatim}
ost pv1 [0] [2 1];
reg pv1s {0,1,2} {0,1,2} "([0,0]|([0,1][1,1][1,0]|[0,1][1,0])|
     [0,2][2,0])*":
def pv1bn "?msd_pv1 Et,u,v n=t+1 & $pv1s(t,u) & $pv1s(u,v) 
   & v=2*z+3*t":
\end{verbatim}
DFAO for $(\sqrt{3}-1)/2$
in base 2:
\begin{verbatim}
def pv1bn_d2 "?msd_pv1 Ex,y $pv1bn(2*n,x) & $pv1bn(n,y) & x!=2*y":
combine PV1B2 pv1bn_d2:
\end{verbatim}
DFAO for $\sqrt{3}+1$
in base 2:
\begin{verbatim}
def pv1an "?msd_pv1 Eu $pv1bn(2*n,u) & z=(u+2*n)/1":
def pv1n_d2 "?msd_pv1 Ex,y $pv1an(2*n,x) & $pv1an(n,y) & x!=2*y":
combine PV12 pv1n_d2:
\end{verbatim}

\subsubsection{Pisot number
$(\sqrt{17}+3)/2$
and 
$(\sqrt{17}-3)/4$
in base $2$}
\label{sec:Pisot2}


\begin{verbatim}
ost pv2 [0] [3 1 1];
reg pv2s {0,1,2,3} {0,1,2,3} 
   "([0,0]|[0,1][1,0]|[0,1][1,1][1,0]|[0,2][2,0]|
   [0,2][2,1][1,0]|[0,3][3,0])*":
def pv2bn "?msd_pv2 Es,t,u,v n=s+1 & $pv2s(s,t) 
   & $pv2s(t,u) & $pv2s(u,v) & v=4*z+7*s":
\end{verbatim}
DFAO for $(\sqrt{17}-3)/4$
in base 2 (see Figure~\ref{figcf0311}):
\begin{verbatim}
def pv2bn_d2 "?msd_pv2 Ex,y $pv2bn(2*n,x) & $pv2bn(n,y) & x!=2*y":
combine PV2B2 pv2bn_d2:
\end{verbatim}
DFAO for $(\sqrt{17}+3)/2$
in base 2:
\begin{verbatim}
def pv2an "?msd_pv2 Eu $pv2bn(2*n,u) & z=(u+3*n)/1":
def pv2n_d2 "?msd_pv2 Ex,y $pv2an(2*n,x) & $pv2an(n,y) & x!=2*y":
combine PV22 pv2n_d2:
\end{verbatim}

\begin{figure}
\begin{center}
\includegraphics[width=6in]{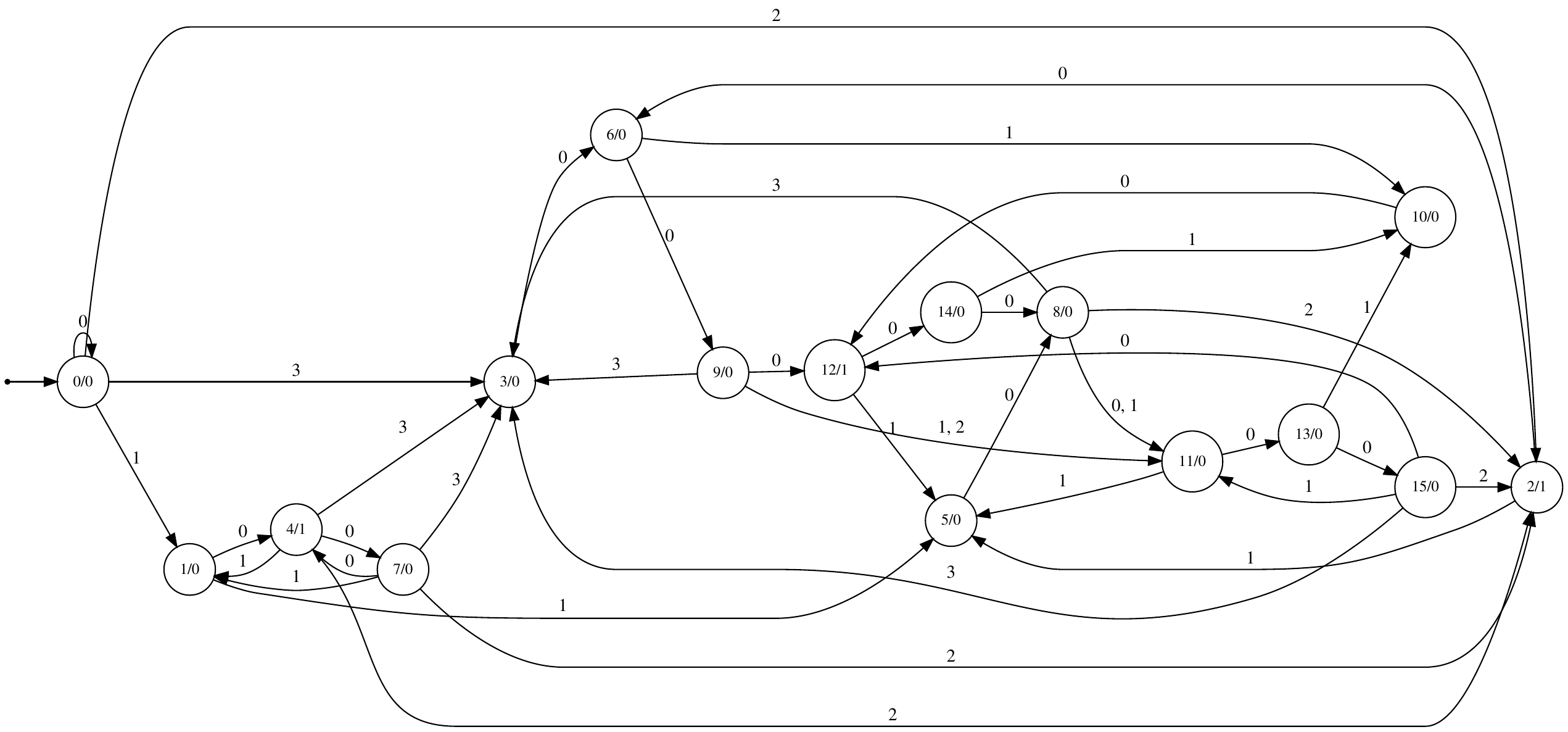}
\end{center}
\caption{Automaton for the $n$'th bit to the right of the point of 
$(\sqrt{17}-3)/4$
in base $2$.}
\label{figcf0311}
\end{figure}


\section{Are the automata minimal?}

The automata that {\tt Walnut} constructs for computing
$ \lfloor bq \varphi \rfloor - b \lfloor q \varphi \rfloor$ on input $q \geq 0$
are guaranteed to be minimal.   However, in this paper, with our application to computing the base-$b$ digits of $\varphi$, we are
only interested in running these automata in the special case when $q = b^i$, the
powers of $b$.   Could it be that there are even smaller automata
that answer correctly on inputs of the form $b^i$ (but might give
a different answer for other inputs)?   After all, for each $t$, we are only concerned with behavior of the automaton on linearly many inputs of length $t$, as opposed to the exponentially large set of valid length-$t$ Zeckendorf representations.   Thus, the automaton is not very constrained.

We do not know the answer to this question, in general.   The question is likely
difficult; in terms of computational complexity,
it is a special case of a problem known to be 
NP-hard, namely, the problem of inferring a minimal DFAO from
incomplete data \cite{Gold:1978}.
However, this problem can sometimes be solved in practice 
using satisfiability (SAT) solving \cite{Zakirzyanov&Shalyto&Ulyantsev:2018}.

We are able to show that some of our automata are indeed minimal, among all automata giving the correct answers on inputs of the
form $q = b^i$, and satisfying two conventions:
first, that  leading zeroes in the input cannot affect the result, and second, that the automata obey the Ostrowski rules for the particular numeration system.
Our method of proving minimality (and in some cases uniqueness) uses SAT solving.

We use a modified version of a MinDFA solver called {\tt DFA-Inductor} \cite{Zakirzyanov&Shalyto&Ulyantsev:2018} to generate SAT encodings for minimal automata, which are then passed to the {\tt CaDiCaL} SAT solver~\cite{cadical} to determine whether they have a satisfying solution.
{\tt DFA-Inductor} uses the \textit{compact encoding} method given by Heule and Verwer~\cite{Heule&Verwer:2010}, which defines eight constraints---four mandatory and four redundant---to translate DFA identification into a graph coloring problem, and then encodes those constraints into a SAT instance.
In short, a set of accepting and rejecting strings from a given dictionary are used to construct an automaton called an \textit{augmented prefix tree acceptor} (APTA), which is then used to construct a \textit{consistency graph} (CG) made up from vertices of the APTA.
Edges in the consistency graph identify the vertices of the APTA that cannot be safely merged together, and by partitioning the vertices as disjoint sets of equivalent states and iterating over the number of partitions, a minimal DFA can be constructed.
Symmetry-breaking predicates are used to enforce a lexicographic breadth-first search (BFS) enumeration on the ordering of states in the constructed DFA, which reduces the size of the search space by removing isomorphic automata from consideration~\cite{Zakirzyanov&Shalyto&Ulyantsev:2018}. 


One of the redundant constraints of the compact encoding method adds all of the determinization conflicts from the consistency graph as binary clauses in the encoding. 
We found it was often the case that the time required to run the determinization step during the generation of the consistency graph was considerably higher than the time required by the solver to find a solution without the extra clauses. Therefore, we excluded the determinization step, so that only direct conflicts between accepting and rejecting states were added as binary clauses.

{\tt DFA-Inductor} only supports DFAs (and hence only accepting or rejecting states), however, and additional output status labels were added for bases larger than 2.
{\tt DFA-Inductor} does not explicitly encode a ``dead state'' rejecting invalid strings, but a transition to a dead state can be implied by a lack of an outgoing transition on a given state.  
Another redundant constraint of the compact encoding method forces each state to have an outgoing transition on every symbol, which is required under the formal definition of a DFA\@.  In order to accommodate the virtual dead state requirement, the constraint must be amended to exclude whichever symbols must transition to the implied dead state.

Our automata follow the convention that the start state consumes leading $0$'s in the input string. In terms of the compact encoding variables, $y_{l,p,q}$ indicates that state $p$ has a transition to state $q$ on label $l$, or in the context of graph coloring, that parents of vertices with color $q$ and incoming label $l$ must have color~$p$.  This constraint is then implemented by enforcing state $0$ to have a self-loop on the symbol $0$ using the unit clause $y_{0,0,0}$, and the dictionary given to {\tt DFA-Inductor} states that the string $0$ produces output $0$.

In order for the SAT solver to construct automata that obey the rules of a given Ostrowski representation, we encode the Ostrowski rules as a set of constraints.
Without these constraints, the solver may find a smaller DFAO by allowing rule-breaking transitions---%
such as allowing consecutive 1's for $\varphi$ in the Zeckendorf representation. 
In Section~\ref{sec:ostEncoding1}, we provide a SAT encoding of the Ostrowski rules for quadratic irrationals with a purely periodic continued fraction with a period of 1, such as $\varphi=[1,\overline1]$, $(\sqrt8+2)/2=[2,\overline2]$, and $(\sqrt{13}+3)/2=[3,\overline3]$.
In Section~\ref{sec:ostEncoding2}, we provide an encoding for purely periodic quadratic irrationals with longer periods, like
$(\sqrt{3}-1)/2=[0,\overline{2,1}]$ and $(\sqrt{17}-3)/4=[0,\overline{3,1,1}]$.

\subsection{A simple Ostrowski encoding for metallic means}
\label{sec:ostEncoding1}

A \emph{metallic mean} is a quadratic irrational $\alpha = [d_1, \overline{d_1}]$ with only a single repeating term $d_1$ in its period. When considering how to encode the three rules for the Ostrowski $\alpha$-representation of $N = [a_{n-1}a_{n-2}\cdots a_0]_\alpha$, we can think of each transition in the DFA as choosing a value for $a_i$, starting with $a_{n-1}$ and working down to $a_0$. Rule 1 requires that $a_0 < d_1$, which is enforced by requiring all strings in the dictionary to be valid strings in the representation, rather than through constraints on the solver. In the context of the metallic means, rule 2 requires that $a_i \leq d_1$ for all $i \geq 1$, which equates to restricting the set of valid transition labels for each state to be in the range of $0$ to $d_1$. This is done automatically during the construction of the APTA, because $d_1$ is the largest value in the dictionary and thus the highest possible transition label.  For $i\geq1$,  rule 3 requires that if $a_i = d_1$, then $a_{i-1} = 0$; this is the only constraint that must be explicitly encoded. Since every state has the same set of allowed transitions, we only need to enforce that no state have a self-loop on label $d_1$, and that if state $q_a$ transitions to $q_b$ on label $d_1$, then $q_b$ must transition to the dead state on labels $1$ to $d_1$. Since dead states are implicit, the constraint is implemented by removing all outgoing transitions from $q_b$ except $0$. The constraints are 
\begin{displaymath}
\bigwedge_{i \in Q} \lnot y_{d_1,i,i} \quad \text{ and } \quad
\bigwedge_{\substack{i, j, k \in Q \\i \neq j }} \bigwedge_{1\leq l\leq d_1}(y_{d_1,i,j} \rightarrow \lnot y_{l,j,k})
\end{displaymath}
where $Q$ is the set of states in the resulting DFAO.

\subsection{Ostrowski encoding for purely periodic quadratic irrationals}
\label{sec:ostEncoding2}

Quadratic irrationals of the form $\alpha = [0, \overline{d_1, d_2, \ldots, d_t} ]$, where $d_1 > 1$, require a more robust encoding than the metallic means in order to account for multiple terms in the period. The order of terms in the continued fraction determines the set of valid transitions between any two states. Therefore, the SAT solver must understand how to relate a given state in the DFAO to a given term in the continued fraction, otherwise the resulting DFAO may fail to reject invalid strings. We now present one such encoding.

Each Ostrowski $\alpha$-representation is a language made up from the set of valid strings that can be constructed using the Ostrowski rules. This language is recognized by a canonical DFA, and serves as the base that informs the valid structure of the final DFAO\@. Constructing a DFAO using only the states in the base DFA guarantees that rule 2 and rule 3 of the Ostrowski construction are never violated. Furthermore, a minimal complete DFAO constructed from only the base states is guaranteed to be minimal for the language. If it were possible to construct an even smaller complete DFAO that is still correct, then a state would exist that is not in the base DFA, implying that at least two unique base states could be safely combined, and thus that the canonical base DFA is not minimal. Therefore, the base states describe the complete set of rules the SAT solver must understand in order to construct a valid DFAO\@. Conveniently, {\tt Walnut} automatically generates a DFA of the Ostrowski base during the process of constructing the representation. 

Since each state in the base DFA has a unique transition set, we can refer to the $i$'th state in the base DFA as the $i$'th \textit{base state}. For example, Figure~\ref{figWal21} shows for $\alpha = (\sqrt{3}-1)/2 = [0, \overline{2,1}]$ how each base state in the Ostrowski base DFA (bottom), labelled {\tt B0} to {\tt B5}, correspond exactly to a state in the DFAO for returning the $i$'th digit of $\alpha$ in base 2 (top). 

\begin{figure}[H]
\begin{center}
\includegraphics[width=6in]{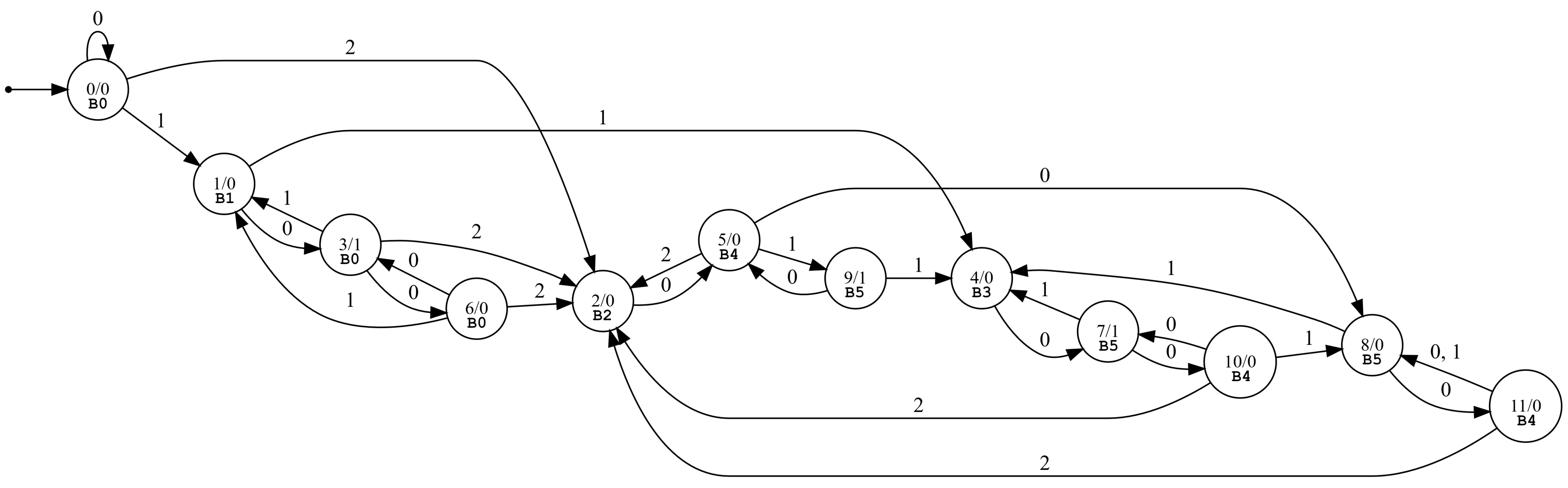}
\includegraphics[width=3in]{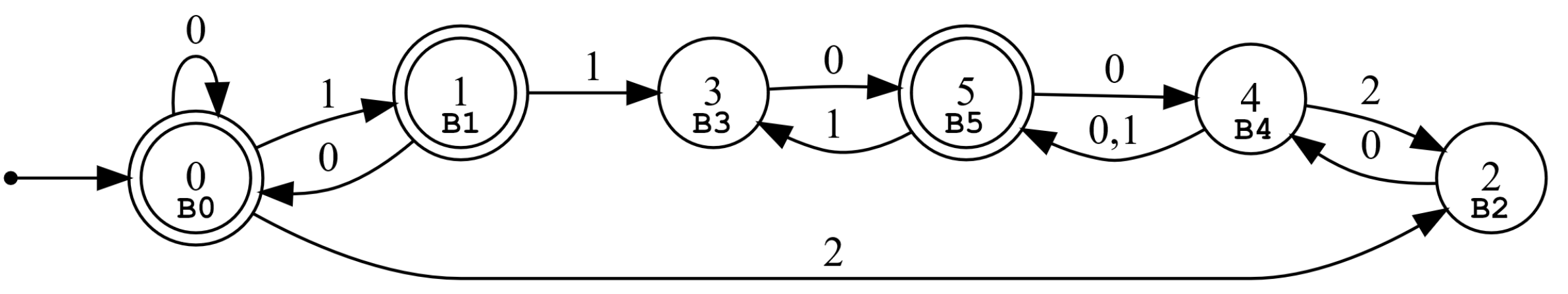}
\end{center}
\caption{Relationship between the base states in the DFA for the Ostrowski representation of $\alpha = (\sqrt{3}-1)/2$ (bottom), labelled {\tt B0} to {\tt B5}, and the states in the DFAO for returning the $n$'th bit to the right of the point of $\alpha$ in base 2 (top).}
\label{figWal21}
\end{figure}

Each cyclic permutation of terms in the continued fraction requires a unique SAT encoding for the states of the Ostrowski base DFA. Term orderings that preserve the permutation differ only in which base states are accepting or rejecting. Since the dictionary file  contains only accepted strings, the SAT solver does not need to know from this encoding which states of the Ostrowski base DFA should be rejecting, and so only the transition set of each base type is encoded. 

The Ostrowski rules are encoded through the states in the Ostrowski base DFA by constraining each state in the DFAO to match a certain base state. Therefore, to encode the base states, we create a new variable $b_{p,t}$, which says state $p$ in the DFAO is related to base state $t$ in the Ostrowski base DFA. We then relate the $b$ variable to the transition variable $y_{l,p,q}$, which constrains the set of valid transitions between $p$ and $q$ according to which base states they are associated with. The encoding is presented in Table~\ref{tableEnc}. 

The last constraint in the table is the only one that needs to be manually determined for each Ostrowski base DFA, depending on the permutation of term orderings in the continued fraction. For example, for $\alpha = (\sqrt{3}-1)/2$ in Figure~\ref{figWal21}, base state {\tt B4} is encoded as follows,
where $Q$ denotes the set of states in the DFAO and $B$ denotes the set of states in the Ostrowski base DFA:
\begin{gather*}
\label{eq:ExWal21_a} \bigwedge_{\substack{i, j \in Q \\i \neq j }}\bigl((b_{i,4} \wedge b_{j,2} \rightarrow \lnot y_{0,i,j}) \wedge (b_{i,4} \wedge b_{j,2} \rightarrow \lnot y_{1,i,j}) \wedge (b_{i,4} \wedge b_{j,5} \rightarrow \lnot y_{2,i,j})\bigr)
\\
\label{eq:ExWal21_b} \bigwedge_{\substack{i, j \in Q \\i \neq j }}\ \bigwedge_{\substack{k \in B \setminus \{2,5\}}}\ \bigwedge_{0\leq l\leq 2}(b_{i,4} \wedge b_{j,k} \rightarrow \lnot y_{l,i,j})
\end{gather*}
\setlength{\tabcolsep}{3pt}
\begin{table}
\centering
\begin{tabular}{m{4cm}>{\raggedright}m{4cm}>{\raggedright\arraybackslash}m{7cm}} 
\thead{\bf Constraints}
& \thead{\bf Range}  
& \thead{\bf Meaning}
\\
\hline
$\lnot y_{k,0,0}$
& ${1 \leq k \leq c}$  
& The start state can only have a self-loop on 0.
\\
\hline
$\lnot y_{k,i,i}$
& ${i \in Q}; i \neq 0; {1 \leq k \leq c}$  
& No states other than the start state can have a self-loop on any label.
\\
\hline
$b_{0,0}$
& $ $
& The start state is related to base state 0.
\\
\hline
${b_{i,s} \rightarrow \lnot b_{i,t}}$    
& {$i \in Q$}; {$s, t \in B$}
& Each state in the DFAO must be related to at most one base type.
\\
\hline
$b_{i,1} \lor b_{i,2} \lor \dotsb \lor b_{i,|B|}$
& ${i \in Q}$
& Each state in the DFAO must be related to at least one base type.     
\\
\hline
$ (b_{i,s} \land b_{j,t}) \rightarrow \lnot y_{k,i,j}$ 
& $i, j \in Q$; $s,t \in B$; $k \in \Sigma$; $\delta(s,k) \neq t$  
& Suppose DFAO state $i$ is related to base state $s$, and state $j$ is related to base state $t$.
If state $s$ in the base DFA does not have a transition to state $t$ on label $k$, then $i$ cannot have a transition to $j$ on label $k$ in the DFAO.
\\
\hline
\end{tabular}%
\caption{SAT encoding of Ostrowski constraints for purely periodic quadratic irrationals.
In this table $\delta$ denotes the transition function, $\Sigma$ denotes the alphabet, and $c = \max(\Sigma)$.}
\label{tableEnc}
\end{table}

\subsection{Results}\label{sec:Results}

Table~\ref{tableWalnut} gives our results of DFA minimization by SAT on a few quadratic irrationals.
In each of the cases, the {\tt Walnut} solution was confirmed to be minimal
by proving that there are no satisfying assignments of the SAT encoding with a smaller number of states
than in the {\tt Walnut}-produced automaton.

The dictionary containing the Ostrowski representation of the first $i$ digits is referred to as the $i$'th digit set.
The solver is run on the SAT encoding of each digit set for a given number of states.
The state count was increased every time the solver returned UNSAT,
and the digit set was increased 
every time a satisfying assignment was found.
Once the state count given by the {\tt Walnut}-produced solution was reached,
the solver was run exhaustively to find all
satisfying assignments of the SAT formula and therefore
\emph{all} candidates for the minimal automata computing the quadratic irrational.
However, most satisfying assignments encoded automata that only computed the
given digit set correctly and did \emph{not} correctly compute the digits
of the quadratic irrational at high precision.

\begin{table}[t]
\centering
\begin{tabular}{ | c | c | c | c | c | c | c | c | }   
\hline 
\thead{Quadratic \\ Irrational}
& \thead{$\varphi$ \\ base 2 \\ 8 states} 
& \thead{$\varphi$ \\ base 3 \\ 13 states}
& \thead{$\sqrt2$ \\ base 2  \\ 6 states }
& \thead{$\frac{\sqrt{13}+3}{2}$ \\ base 2  \\ 7 states }
& \thead{$\frac{\sqrt{13}+3}{2}$ \\ base 3  \\ 8 states }
& \thead{$\frac{\sqrt{3}-1}{2}$ \\ base 2  \\ 12 states }
& \thead{$\frac{\sqrt{17}-3}{4}$ \\ base 2 \\ 16 states}\\
\hline
Digit set size     & 54          & 197       & 29        & 64        & 64        & 27        & 57       \\
SAT time (sec)    & 0.50    & 28,425.5   & 0.08  & 142.81   & 44.68   & 0.14 &  68.11   \\
UNSAT time (sec)  & 0.18        & 12,123.0  & 0.02      & \phantom{00}0.52      & 24.76     & 0.08   &    \phantom{0}2.59      \\
Number of candidates        & 1           & 3         & 1         & 3         & 7         & 1         & 9         \\
\hline
\end{tabular} 
\caption{Results for computing minimal automata for various quadratic irrationals.}
\label{tableWalnut}
\end{table}

The digit set size given in Table~\ref{tableWalnut} is the smallest dictionary required for the SAT solver to find the $n$-state {\tt Walnut} solution.
The SAT time is the time required by the solver to find the {\tt Walnut} automaton.
The UNSAT time is the time required determine no automata is possible using $n-1$ states.
Since no candidate solutions are found at $n-1$ states,
we conclude that the $n$-state {\tt Walnut} solution is minimal.

In some cases, multiple distinct candidates were found
that correctly computed at least 10,000 digits of the quadratic irrational (see the last row of Table~\ref{tableWalnut}).
For reference, the {\tt Walnut} solutions for the cases with multiple candidate solutions are given in Section~\ref{sec:Walnut-extra}.
For all except $(\sqrt{17}-3)/4$, these candidate solutions differ from the {\tt Walnut} solution only by their outgoing transitions on the start state.
The candidates for $\varphi$ (base 3) and $(\sqrt{13}+3)/2$ (base 2) have differing transitions on label 1,
while the candidates for $(\sqrt{13}+3)/2$ (base 3) differ on label 2.
All of the candidates for $(\sqrt{17}-3)/4$ have the same start state, but differ in their transitions on label 2.
Given how similar the 
candidate solutions are to the {\tt Walnut} solution and that they are correct up to a high precision,
it is possible that the {\tt Walnut} solution is not unique, though we leave this as an open problem.

Minimization of DFAOs for this purpose presents a particular challenge for the SAT solver, as both the size of the digit set required to find 
a candidate solution
and the length of the representation for each digit position can be arbitrarily large. 
For this reason, $\varphi$ in base~4 and $\sqrt{2}$ in base 3 encountered prohibitively long solving times before the required number of states (22 states and 14 states, respectively) could be reached, preventing the minimality of the {\tt Walnut} solutions from being determined.
For $\varphi$ in base 4, it took over 25 hours for the 78'th digit set to be declared UNSAT at 13 states, and for $\sqrt{2}$ in base 3, it
took over 55 hours for the 258'th digit set to be declared SAT at 11 states, but the satisfying assignment
found by the solver corresponded to an automata that incorrectly computed the ternary digits of $\sqrt2$ starting at the 321'th digit.


\appendix

\section{Appendix}

In this Appendix we give an automaton that computes the
digits of $\varphi$ in base 10.
This DFAO has transition function $\delta(q,i)$ and output
function $\tau(q)$, with initial state $0$ and state set
$\{0,1,\ldots, 96 \}$.   The transition function and output are
given in Table~\ref{tab6}.

This was generated with the following {\tt Walnut} code:

{\small
\begin{verbatim}
reg shift {0,1} {0,1} "([0,0]|[0,1][1,1]*[1,0])*":
def phin "?msd_fib (s=0 & n=0) | Ex $shift(n-1,x) & s=x+1":
def fibdigit "?msd_fib Ex,y $phin(10*n,x) & $phin(n,y) & z+10*y=x":

def fibd1 "?msd_fib $fibdigit(n,1)":
def fibd2 "?msd_fib $fibdigit(n,2)":
def fibd3 "?msd_fib $fibdigit(n,3)":
def fibd4 "?msd_fib $fibdigit(n,4)":
def fibd5 "?msd_fib $fibdigit(n,5)":
def fibd6 "?msd_fib $fibdigit(n,6)":
def fibd7 "?msd_fib $fibdigit(n,7)":
def fibd8 "?msd_fib $fibdigit(n,8)":
def fibd9 "?msd_fib $fibdigit(n,9)":

combine FD10 fibd1 fibd2 fibd3 fibd4 fibd5 fibd6 fibd7 fibd8 fibd9:

# FD10[10^n] gives the n'th digit in the decimal
# representation of phi = 1.61803...
\end{verbatim}
}

\begin{table}
   \centering
\resizebox{.4\columnwidth}{!}{%
    \begin{tabular}{c|c|c|c||c|c|c|c}
    $q$ & $\delta(q,0)$ & $\delta(q,1)$ & $\tau(q) $ & $q$ & $\delta(q,0)$ & $\delta(q,1)$ & $\tau(q) $ \\
    \hline
  0& 0& 1& 0&49&46&27& 9\\
  1& 2&---& 6&50&71&---& 5\\
  2& 3& 4& 2&51&72&30& 1\\
  3& 5& 6& 8&52&73&---& 7\\
  4& 7&---& 4&53&32&74& 3\\
  5& 8& 9& 0&54&75& 1& 9\\
  6&10&---& 7&55&35&---& 6\\
  7&11&12& 3&56&76& 6& 8\\
  8&13&14& 9&57&77&---& 7\\
  9&15&---& 5&58&78&12& 3\\
 10&16&17& 1&59&79&47& 9\\
 11&18&19& 7&60&80&17& 1\\
 12&20&---& 4&61&51&19& 7\\
 13&21&22& 0&62&81&55& 0\\
 14&23&---& 6&63&24& 4& 2\\
 15&24&25& 2&64&26&82& 8\\
 16&26&27& 8&65&83&---& 4\\
 17&28&---& 5&66&29& 9& 1\\
 18&29&30& 1&67&84&---& 7\\
 19&31&---& 7&68&61&85& 3\\
 20&32&33& 3&69&86&---& 5\\
 21&34&14& 9&70&87&37& 2\\
 22&35&---& 5&71&88&45& 2\\
 23&36&37& 2&72&46&89& 9\\
 24&38&39& 8&73&72&50& 1\\
 25&40& 0& 4&74&90&---& 3\\
 26&41&42& 0&75&54& 1& 0\\
 27&43&---& 6&76&91& 9& 0\\
 28&44&45& 3&77&16&92& 1\\
 29&46&47& 9&78&18&19& 8\\
 30&48&---& 5&79&41&22& 0\\
 31&49&50& 1&80&46&27& 8\\
 32&51&52& 7&81&34& 1& 9\\
 33&53&---& 3&82&77&---& 6\\
 34&54&55& 0&83&44&12& 3\\
 35&56& 4& 2&84&80&50& 1\\
 36& 5&57& 8&85&53&---& 4\\
 37&58&---& 4&86&93& 4& 2\\
 38&59& 9& 1&87&94&82& 8\\
 39&60&---& 7&88&38&67& 8\\
 40&61&12& 3&89&95&---& 6\\
 41&62&14& 9&90&96&74& 3\\
 42&63&---& 5&91&13&47& 9\\
 43&64&65& 1&92&83&---& 5\\
 44&66&67& 8&93&76&39& 8\\
 45&68&---& 4&94& 8&42& 0\\
 46&41&69& 0&95&87&65& 2\\
 47&70&---& 6&96&73&52& 7\\
 48&24&25& 2\\
    \end{tabular}
    }
    \caption{Transitions and outputs for a DFAO that on input $(10^n)_F$ returns
the $n$'th digit to the right of the decimal point in the
decimal expansion of $\varphi$.}
    \label{tab6}
\end{table}


\end{document}